\documentclass[12pt]{article}
\usepackage[title]{appendix}
\usepackage[edges]{forest}
\usepackage{amsfonts}
\usepackage{mathtools}
\usepackage{subcaption}
\usepackage{graphicx}
\usepackage{tabularx}
\usepackage[nomessages]{fp}

\usepackage{enumitem} 
\setlist{  
  listparindent=\parindent,
  parsep=0pt,
}

\usepackage{amsmath,amssymb,epsf,epsfig,amsthm,times,ifthen,epic,psfrag,datetime}
\usepackage{multirow,makecell,tikz}
\usepackage{pgfplots,xfp}
\usepackage{cite}     
\usepackage{fancyhdr}
\usepackage{float}    
\usepackage{setspace} 
\usepackage{lastpage} 

\usetikzlibrary{math}

\def\submitteddate{April 9, 2024}

\renewcommand{\baselinestretch}{1}
\newcommand{\NOP}[1]{}

\newcommand{\notimplies}{%
  \mathrel{{\ooalign{\hidewidth$\not\phantom{=}$\hidewidth\cr$\implies$}}}}

\newtheorem{theorem}              {Theorem}     [section]

\newtheorem{lemma}      [theorem] {Lemma}

\newtheorem{example}    [theorem] {Example}

\theoremstyle{definition}         
\newtheorem{definition} [theorem] {Definition}

\newcommand{\Hide}[1]{}

\DeclareFontFamily{U}{matha}{\hyphenchar\font45}
\DeclareFontShape{U}{matha}{m}{n}{
      <5> <6> <7> <8> <9> <10> gen * matha
      <10.95> matha10 <12> <14.4> <17.28> <20.74> <24.88> matha12
      }{}
\DeclareSymbolFont{matha}{U}{matha}{m}{n}
\DeclareMathSymbol{\notdivides}{3}{matha}{"1F}
\DeclareMathSymbol{\divides}{3}{matha}{"17}

\usepackage{hyperref}
\hypersetup{
    colorlinks=true,
    citecolor=black,
    filecolor=black,
    linktoc=all,     
    linkcolor=blue,
    pdfpagemode=UseNone, 
    urlcolor=black
}

\tikzset{
  iv/.style={
    draw,
    fill=orange!50,
    rectangle,
    minimum size=20pt,
    inner sep=0pt,
    text=black},
  ev/.style={
    draw,
    fill=green,
    rectangle,
    minimum size=20pt,
    inner sep=0pt,
    text=black}}

\begin{document}

\newcommand{\creationtime}{\today\ @ \currenttime}

\pagestyle{fancy}
\renewcommand{\headrulewidth}{0cm}
\chead{\footnotesize{Congero-Zeger}}
\rhead{\footnotesize{\submitteddate}}
\lhead{\footnotesize{\textit{A characterization of optimal prefix codes}}}
\cfoot{Page \arabic{page} of \pageref{LastPage}} 

\makeatletter\renewcommand{\@seccntformat}[1]{\noindent {\csname the#1\endcsname}.\hspace{0.5em}}\makeatother

\renewcommand{\qedsymbol}{$\blacksquare$} 

\newcommand{\Alphabet}{S}

\setcounter{page}{1}

\title{A Characterization of Optimal Prefix Codes
\thanks{
   \indent \textbf{S. Congero and K. Zeger} are with the 
  Department of Electrical and Computer Engineering, 
  University of California, San Diego, 
  La Jolla, CA 92093-0407 
  (scongero@ucsd.edu and ken@zeger.us).
}}

\author{Spencer Congero and Kenneth Zeger\\}
\date{
  \textit{
   SIAM Journal on Discrete Mathematics\\
  Submitted: \submitteddate\\
  %
  }
}

\maketitle
\begin{abstract}
A property of prefix codes called strong monotonicity is introduced,
and it is proven that for a given source,
a prefix code is optimal if and only if
it is complete and strongly monotone.
\end{abstract}

\section{Introduction}
\label{sec:introduction}

This paper concerns variable-length binary codes used to transmit or store
source symbols generated by a finite probability distribution.
Our main result is the following characterization of
binary prefix codes which achieve the minimal possible
average codeword length for
a given probability distribution over a finite symbol set.

\begin{theorem}
A prefix code is optimal if and only if
it is complete and strongly monotone.
\label{thm:strongly-monotone}
\end{theorem}

In what follows,
we first give historical background,
then define terminology,
and finally prove the main result.


Huffman codes were invented in 
1952~\cite{Huffman-1952} 
and today are widely used in many practical data compression applications,
such as for 
text, audio, image, and video coding.
They are known to be optimal in the sense that they achieve
the minimal possible expected codeword length among all prefix codes for
a given finite discrete random source~\cite{Cover-Thomas-book-2006}.

The main idea in the Huffman algorithm
is to construct a binary code tree from a finite source
by recursively merging two smallest-probability nodes
until only one node with probability $1$ remains.
The initial source probabilities correspond to leaf nodes in the tree,
and the binary paths from the tree's root to the leaves are the codewords.

For a given source,
Huffman codes and their corresponding code trees are not generally unique,
due to choices that arise during the tree construction
that can be decided arbitrarily:
(1) When two nodes are merged, the choice of which node becomes a left child and
which becomes a right child is arbitrary;
(2) If there are three or more smallest-probability nodes,
then which two of them to merge is arbitrary;
and
(3) If there is a unique smallest-probability node and
two or more second-smallest-probability nodes,
then which of these to merge with the smallest-probability node is arbitrary.
These latter two cases do not occur if 
probability ``ties'' are absent among tree nodes, 
which is almost surely true if the source itself is randomly chosen from
a continuous distribution.
After the tree is constructed 
all edges from parents to left children are labeled $0$
and 
all edges from parents to right children are labeled $1$,
or vice versa.

For many applications,
the average length of a prefix code is a primary concern,
while in some applications,
the specific binary codewords included 
in an optimal code may also be critical,
such as for reducing average resynchronization time
when channel errors can occur
(e.g., 
\cite{Cao-Yao-Chen-2007},
\cite{Escott-Perkins-1998},
\cite{Ferguson-Rabinowitz-1984},
\cite{Freiling-Jungreis-Theberge-Zeger},
\cite{Higgs-Perkins-Smith-2009},
\cite{Longo-Galasso-1982},
\cite{Zhou-Zhang-2002},
\cite{Zhou-Au-2010}
).

In addition to there being multiple Huffman codes for a given source,
there are also generally multiple non-Huffman codes which achieve the
same minimal average codeword length as Huffman codes.
These optimal non-Huffman codes can have different codewords and
their code trees may be topologically different.

Mathematical and algorithmic 
characterizations of Huffman codes,
and more generally optimal prefix codes, 
have been of great interest over the last 70 years.

Algorithmically, 
code equivalences have previously been described in terms of
various ``node swap''  transformations of the corresponding code trees.
``Same-parent'' node swaps consist of switching the two siblings 
(and the entire subtrees hanging from them)
of a parent node in the tree.
Similarly,
``same-row'' node swaps switch two nodes in the same tree row,
and ``same-probability'' node swaps
switch two tree nodes having the same probability.
It can be shown
that any two complete prefix codes that are length equivalent
can be obtained from each other by a series of same-row 
node swaps. 
Also, any two Huffman codes for the same source
can be obtained from each other by a series of same-parent and same-probability
node swaps~\cite{Longo-Galasso-1982}.
Additionally, every optimal code is 
length equivalent to some Huffman code~\cite{Manickman-2019}.
These two results of 
~\cite{Longo-Galasso-1982} and ~\cite{Manickman-2019}
together imply that
any two optimal codes 
can be obtained from each other by a series of same-row and same-probability
node swaps.

In 1978, Gallager~\cite{Gallager-IT-1978}
gave a useful non-algorithmic characterization of Huffman codes 
as precisely those prefix codes 
possessing a ``sibling property'',
which stipulates that a code is complete and
the nodes of its code tree
can be listed in order of non-increasing probability with
each node being adjacent in the list to its sibling.

For the broader class of optimal prefix codes,
no characterization analogous to the sibling property
has been previously given,
and it has remained an open question until present.
Only the sufficient condition given by the sibling property has been known.

One known necessary condition for a prefix code to be optimal is ``monotonicity'',
which states generally that code tree node probabilities decrease 
moving downward in the code tree.

An interpretation of monotonicity is 
as a constraint on the ordering of the collections
of each tree node's full set of leaf descendants,
where we note that each collection has a Kraft sum which is an integral power of $1/2$.
In this paper we provide a necessary and sufficient characterization
of optimal prefix codes
by introducing a new criterion called ``strong monotonicity''.
Strong monotonicity generalizes the above interpretation of monotonicity 
by removing the restriction
that each leaf collection comprises the descendants of a single tree node.

In particular we show that for a given source,
a prefix code is optimal
if and only if it is complete and strongly monotone.
We also note that 
Theorem~\ref{thm:strongly-monotone}
is exploited in 
another recent work~\cite{Congero-Zeger-Archiv-competitive}
to prove results about competitive optimality of Huffman codes.

Figure~\ref{fig:block-diagram} depicts the main result
relating to 
Theorem~\ref{thm:strongly-monotone},
along with some known prior art.
%
%
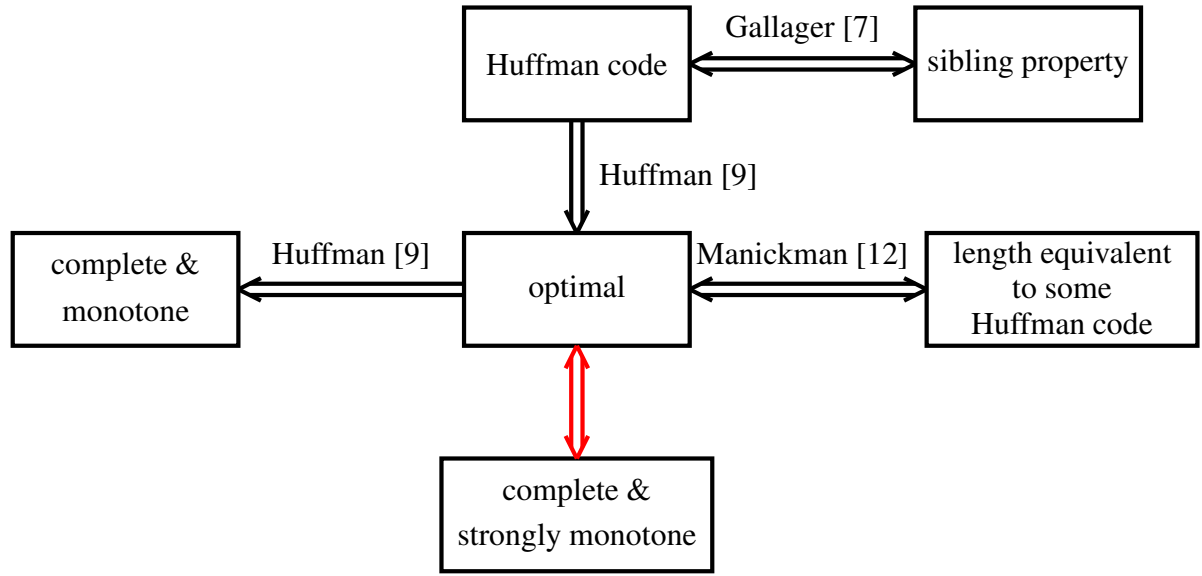
\begin{figure}[h]
\begin{center}
\begin{tikzpicture}[scale=1.5]
\draw[-, ultra thick]  (0,-2) -- (2,-2) -- (2,-1) -- (0,-1) -- cycle;     
\draw[-, ultra thick]  (4,0) -- (6,0) -- (6,1) -- (4,1) -- cycle;         
\draw[-, ultra thick]  (8,0) -- (10,0) -- (10,1) -- (8,1) -- cycle;       
\draw[-, ultra thick] (3.8,-4)--(6.2,-4)--(6.2,-3) --(3.8,-3) -- cycle;   
\draw[-, ultra thick] (4,-2)--(6,-2)--(6,-1) --(4,-1) -- cycle;           
\draw[-, ultra thick] (8.1,-2)--(10.5,-2)--(10.5,-1) --(8.1,-1) -- cycle; 

\draw[-, ultra thick]  (6.1, 0.55) -- (7.9, 0.55); 
\draw[-, ultra thick]  (6.1, 0.45) -- (7.9, 0.45); 
\draw[-, ultra thick]  (8,     .5) -- (7.8,  0.6); 
\draw[-, ultra thick]  (8,     .5) -- (7.8,  0.4); 
\draw[-, ultra thick]  (6.0, 0.50) -- (6.2, 0.40); 
\draw[-, ultra thick]  (6.0, 0.50) -- (6.2, 0.60);

\draw[-, ultra thick]  (4.95,    0) -- (4.95, -0.90); 
\draw[-, ultra thick]  (5.05,    0) -- (5.05, -0.90); 
\draw[-, ultra thick]  (5.0,  -1.0) -- (4.9,  -0.80); 
\draw[-, ultra thick]  (5.0,  -1.0) -- (5.1,  -0.80);

\draw[-, red, ultra thick]  (4.95,   -2.1) -- (4.95, -2.90); 
\draw[-, red, ultra thick]  (5.05,   -2.1) -- (5.05, -2.90); 
\draw[-, red, ultra thick]  (5.0,  -3.0) -- (4.9,  -2.80); 
\draw[-, red, ultra thick]  (5.0,  -3.0) -- (5.1,  -2.80); 
\draw[-, red, ultra thick]  (5.0,  -2.0) -- (4.9,  -2.2); 
\draw[-, red, ultra thick]  (5.0,  -2.0) -- (5.1,  -2.2);

\draw[-, ultra thick]  (2.1, -1.55) -- (4, -1.55); 
\draw[-, ultra thick]  (2.1, -1.45) -- (4, -1.45); 
\draw[-, ultra thick]  (2.0, -1.50) -- (2.2, -1.60); 
\draw[-, ultra thick]  (2.0, -1.50) -- (2.2, -1.40);

\draw[-, ultra thick]  (6.1, -1.55) -- (8, -1.55); 
\draw[-, ultra thick]  (6.1, -1.45) -- (8, -1.45); 
\draw[-, ultra thick]  (6,   -1.50) -- (6.2, -1.60); 
\draw[-, ultra thick]  (6,   -1.50) -- (6.2, -1.40); 
\draw[-, ultra thick]  (8.1, -1.50) -- (7.9, -1.60); 
\draw[-, ultra thick]  (8.1, -1.50) -- (7.9, -1.40); 

\node[black, scale=1] at (1,  -1.3) {complete \&};
\node[black, scale=1] at (1,  -1.7) {monotone};
\node[black, scale=1] at (5,  0.5) {Huffman code};
\node[black, scale=1] at (9,  0.5) {sibling property};
\node[black, scale=1] at (5, -3.3) {complete \&};
\node[black, scale=1] at (5, -3.7) {strongly monotone};
\node[black, scale=1] at (5, -1.5) {optimal};
\node[black, scale=1] at (9.3, -1.2) {length equivalent};
\node[black, scale=1] at (9.3, -1.5) {to some};
\node[black, scale=1] at (9.3, -1.8) {Huffman code};
\node[black, scale=1] at (3,  -1.2) {Huffman\cite{Huffman-1952}};
\node[black, scale=1] at (5.9, -0.5) {Huffman\cite{Huffman-1952}};
\node[black, scale=1] at (7,  0.8) {Gallager\cite{Gallager-IT-1978}};
\node[black, scale=1] at (7,  -1.2) {Manickman\cite{Manickman-2019}};
\end{tikzpicture}
\end{center}
\caption{
Logical implications of prefix code properties for a given source.
The red arrows indicate new results presented in this paper.
}
\label{fig:block-diagram}
\end{figure}


\section{Definitions}
\label{sec:defs}

An \textit{alphabet} is
a finite set $\Alphabet$,
and a \textit{source} with alphabet $\Alphabet$ 
is a random variable $X$ which,
for each $y\in\Alphabet$,
takes on the value $y$
with probability $P(y)$.
The probability of any subset $B\subseteq \Alphabet$ is denoted by
$P(B) = \displaystyle\sum_{y \in B} P(y)$.
A \textit{code} 
for source $X$ is a mapping 
$C:\Alphabet \longrightarrow \{0,1\}^*$
and,
for each $y\in \Alphabet$,
the binary string $C(y)$ 
is a \textit{codeword} of $C$.
A \textit{prefix code} is a code
where no codeword is a prefix of any other codeword.

A \textit{code tree} for a prefix code $C$ is a rooted binary tree 
whose leaves correspond to the codewords of $C$.
By convention,
each edge leading to a left child (respectively, right child) will be 
labeled $0$ (respectively, $1$).
The codeword associated with each leaf
is the binary word describing the path from the root to the leaf.
The \textit{length} of a code tree node is its path length from the root.
The $r$th \textit{row} of a code tree is the set of nodes whose length is $r$,
and we will view a code tree's root as being on the top of the tree
with the tree growing downward.
That is, row $r$ of a code tree is ``higher'' in the tree than row $r+1$.
If $x$ and $y$ are nodes in a code tree,
then $x$ is a \textit{descendant} of $y$
if there is a downward 
path of length zero or more from $y$ to $x$.
Two nodes in a tree are called \textit{siblings} if they have the same parent.
For any collection $A$ of nodes in a code tree,
let $P(A)$ denote the probability
of the set of all leaf descendants of $A$ in the tree.

A (binary) \textit{Huffman tree} is a code tree constructed from a source
by recursively merging two smallest-probability nodes%
\footnote{
For more details about Huffman codes, 
the reader is referred to the textbook
~\cite[Section 5.6]{Cover-Thomas-book-2006}.
}
until only
one node with probability $1$ remains.
The initial source probabilities correspond to leaf nodes in the tree.
A \textit{Huffman code} for a given source is a mapping of source symbols to
binary words by assigning the source symbol corresponding to each leaf
in the Huffman tree to the binary word describing the path from the root
to that leaf.

Given a source with alphabet $\Alphabet$ and a prefix code $C$,
for each $y\in \Alphabet$
the length of the binary codeword $C(y)$ is denoted $l_C(y)$.
Two codes $C_1$ and $C_2$ are \textit{length equivalent} if
$l_{C_1}(y) = l_{C_2}(y)$ for every source symbol $y\in\Alphabet$.
The \textit{average length} of a code $C$ for a source with alphabet $\Alphabet$ is
$\displaystyle\sum_{y\in\Alphabet} l_C(y) P(y)$.
A prefix code is \textit{optimal}
for a given source if no other prefix 
code achieves a smaller average codeword length for the source.
In particular,
Huffman codes are known to be optimal
(e.g., see~\cite{Cover-Thomas-book-2006}).

A code is \textit{complete} if every non-root node in its code tree has a sibling,
or, equivalently, if every node has either zero or two children.%
\footnote{
Our usage of the word ``complete'' has also been referred to in the literature as
``full'',
``extended'',
``saturated'',
``exhaustive'',
and
``maximal''.
}
A code $C$ for a given source
is \textit{monotone} if
for any two nodes in the code tree of $C$,
we have $P(u)\ge P(v)$
whenever $l_C(u) < l_C(v)$.
These two conditions are necessary for optimality,
as stated in the following lemma.

\begin{lemma}[{Huffman~\cite[p. 1099]{Huffman-1952}}]\ \\
For any source,
if a prefix code is optimal,
then it is complete and monotone.
\label{lem:minlen-implies-monotone}
\end{lemma}

The \textit{Kraft sum}~\cite{Kraft-1949} of  a sequence of nonnegative integers
$l_1, \dots, l_k$ is 
$2^{-l_1} + \dots + 2^{-l_k}$.
We extend the definition of ``Kraft sum'' to sets of
source symbols with respect to a code as follows.
If $C$ is a prefix code for a source with alphabet $\Alphabet$,
and $U\subseteq \Alphabet$, 
then the Kraft sum of $U$ is
\begin{align*}
K_C(U) &= \sum_{x\in U} 2^{-l_C(x)} .
\end{align*}

The following lemma is a standard result in most information theory textbooks
and is used 
in the proofs of
Lemma~\ref{lem:complete_KS_1} and Theorem~\ref{thm:strongly-monotone}.
\begin{lemma}[Kraft Inequality converse 
              {\cite[Theorem 5.2.1]{Cover-Thomas-book-2006}}]\ \\ 
If a sequence $l_1, \dots, l_n$ of positive integers satisfies
$2^{-l_1} + \dots + 2^{-l_n} \le 1$,
then there exists a binary prefix code whose codeword lengths are
$l_1, \dots, l_n$.
\label{lem:Kraft-Inequality}
\end{lemma}

The following property was introduced by Gallager and appears in the following
Lemma~\ref{lem:SiblingProperty-iff-Huffman}.

\begin{definition} 
A binary code tree has the \textit{sibling property} if
it is complete
and if the nodes
can be listed in order of non-increasing probability with
each node being adjacent in the list to its sibling.
\end{definition}

The next lemma is very useful in proving results about Huffman codes,
and will be exploited in the proof of
Theorem~\ref{thm:strongly-monotone}.

\begin{lemma}[{Gallager~\cite[Theorem 1]{Gallager-IT-1978}}]\ \\
For any source,
a prefix code is a Huffman code
if and only if its code tree has the sibling property.
\label{lem:SiblingProperty-iff-Huffman}
\end{lemma}


\section{Characterization of optimal prefix codes}
\label{sec:Main-Result}

In this section
we prove Theorem~\ref{thm:strongly-monotone},
which gives a new characterization of optimal prefix codes
for a given source.

While all Huffman codes are optimal and were
characterized by Gallager in terms of the sibling property,
not all optimal codes are Huffman codes.
Theorem~\ref{thm:strongly-monotone}
shows that a prefix code is optimal
if and only if it is complete and strongly monotone.
The combination of completeness
and strong monotonicity is weaker than the sibling property,
and thus 
a broader class of prefix codes 
(namely, the optimal ones)
satisfies this combination.

\begin{definition}
Given a source with alphabet $\Alphabet$,
a prefix code $C$ is \textit{strongly monotone}
if for every $A, B \subseteq S$
with $K_C(A) = 2^{-i} > 2^{-j} = K_C(B)$
for some integers $i$ and $j$,
we have
$P(A) \ge P(B)$.
\end{definition}

The strongly monotone property reduces to Gallager's monotone property
when the set $A$ consists of all leaf descendants of a single tree node,
and $B$ consists of all leaf descendants of a different tree node.
Example~\ref{ex:strong-monotone} 
illustrates that these two properties are not equivalent.
Specifically, the example shows that 
prefix code $C$ is not strongly monotone
because $K_C(\{c,d\}) = 2^{-1} > 2^{-2} = K_C(\{a\})$
but $P(\{c,d\}) = \frac{1}{4} < \frac{3}{8} = P(\{a\})$.

Monotonicity and completeness do not imply
strong monotonicity, nor do they imply optimality,
as illustrated in the following simple example.
\begin{example}[complete \& monotone $\notimplies$ optimal]\ 
See Figure~\ref{fig:monotone-but-not-strongly-monotone}.\\
A balanced prefix code $C$ 
for a source with symbols $a$,$b$,$c$,$d$,
and probabilities 
$\frac{3}{8},\frac{3}{8},\frac{1}{8},\frac{1}{8}$, respectively,
is complete, monotone,
and has average length $2$.
But $C$ is not optimal since a Huffman code has codewords of lengths 
$1,2,3,3$ and smaller average length $15/8$.

\tikzset{every label/.style={xshift=0ex, text width=6ex, align=center, yshift=-6ex,
                             inner sep=1pt, font=\footnotesize, text=red}}
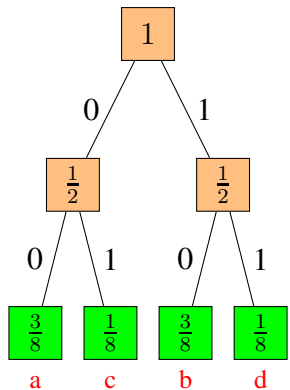
\begin{figure}[H]
\begin{center}
\begin{forest}
for tree={where n children={0}{ev}{iv},l+=8mm,
if n=1{edge label={node [midway, left] {0} } }{edge label={node [midway, right] {1} } },}
      [$1$, baseline
         [$\frac{1}{2}$,
           [$\frac{3}{8}$, label=a]
           [$\frac{1}{8}$, label=c]]
         [$\frac{1}{2}$,
           [$\frac{3}{8}$, label=b]
           [$\frac{1}{8}$, label=d]]]
\end{forest}
\end{center}
\caption{
A code tree illustrating monotonicity without strong monotonicity.
}
\label{fig:monotone-but-not-strongly-monotone}
\end{figure}
\label{ex:strong-monotone}
\end{example}

Strong monotonicity of a prefix code does not imply the code is complete,
and hence the code may not be optimal,
as illustrated in the following simple example.

\begin{example}[strongly monotone $\notimplies$ optimal]\ 
See Figure~\ref{fig:strongly-monotone-not-complete}.\\

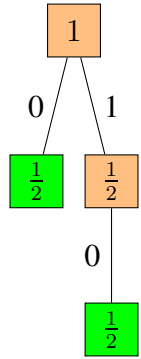
\begin{figure}[H] 
\begin{center}
\begin{forest}
for tree={where n children={0}{ev}{iv},l+=8mm,
if n=1{edge label={node [midway, left] {0} } }{edge label={node [midway, right] {1} } },}
      [$1$, baseline
         [$\frac{1}{2}$]
         [$\frac{1}{2}$,
           [$\frac{1}{2}$]]]
\end{forest}
\end{center}
\caption{
An non-complete prefix code tree
illustrating strong monotonicity without optimality.
}
\label{fig:strongly-monotone-not-complete}
\end{figure}
\label{ex:not-complete}
\end{example}

The following lemma easily follows from the proof of 
Lemma~\ref{lem:Kraft-Inequality}.
This lemma relies on our defining assumption that sources (and thus codes) are finite.
Prefix codes for infinite sources need not satisfy the lemma below.

\begin{lemma}
A prefix code is complete
if and only if
for every node $u$ in its code tree,
the Kraft sum of the set of leaf descendants of $u$ equals $2^{-i}$
where $i$ is the length of $u$.
Also,
a prefix code is complete if and only if its Kraft sum equals $1$.
\label{lem:complete_KS_1}
\end{lemma}

The following lemma lists properties that do not change 
among length equivalent prefix codes.

\begin{lemma}
If two prefix codes are length equivalent,
then each of the following properties
holds for one code
if and only if
it holds for the other code:
(1) completeness;
(2) strong monotonicity;
(3) optimality.
\label{lem:length_equivalent_properties}
\end{lemma}

\begin{proof}
Let $\Alphabet$ be the source alphabet.
Let $C$ and $C'$ be
length equivalent prefix codes,
i.e., $l_C(y) = l_{C'}(y)$ for all $y \in \Alphabet$.
Then for all $y \in \Alphabet$,
\begin{align*}
K_{C}(\{y\})
&= 2^{-l_C(y)}
= 2^{-l_{C'}(y)}
= K_{C'}(\{y\}).
\end{align*}

Since
\begin{align*}
\sum_{y \in \Alphabet} K_C(\{y\})
&= \sum_{y \in \Alphabet} K_{C'}(\{y\}),
\end{align*}
$K_C(S)=1$ if and only if $K_{C'}(S)=1$,
so Lemma~\ref{lem:complete_KS_1} 
implies that
$C$ is complete if and only if $C'$ is complete.

Suppose $C$ is strongly monotone.
Let $A, B \subseteq \Alphabet$
with $K_{C'}(A) = 2^{-i}$ and $K_{C'}(B) = 2^{-j}$
for some integers $i,j$ such that $0 \le i < j$.
Since $K_C(A) = K_{C'}(A) = 2^{-i}$
and $K_C(B) = K_{C'}(B) = 2^{-j}$,
we have $P(A) \ge P(B)$
since $C$ is strongly monotone.
Thus $C'$ is also strongly monotone.

Let $X$ be a source random variable.
The average length of code $C$ is
\begin{align*}
E[l_C(X)]
&= \sum_{y \in \Alphabet} P(y) l_C(y)
= \sum_{y \in \Alphabet} P(y) l_{C'}(y)
= E[l_{C'}(X)],
\end{align*}
so
$C$ is optimal
if and only if $C'$ is optimal.
\end{proof}

The following proves our main result.

\begin{proof}[Proof of Theorem~\ref{thm:strongly-monotone}]
Let $S$ be the source alphabet,
and let $X$ be the source random variable on $S$.
Define $P(u) = P(X = u)$ for all $u \in S$.

First, suppose $C$ is an optimal prefix code for $X$.
Then $C$ is complete by Lemma~\ref{lem:minlen-implies-monotone}.

Suppose for contradiction that $C$ is not strongly monotone.
Then there exist subsets $A,B \subseteq S$
such that $K_C(A) = 2^{-i} > 2^{-j} = K_C(B)$
for some integers $i$ and $j$,
but $P(A) < P(B)$.
Define a new prefix code $C'$ such that for all $u \in S$,
\begin{align*}
l_{C'}(u)
&= l_C(u) +
\begin{cases}
j-i & \text{if}\ u \in A - B \\
i-j & \text{if}\ u \in B - A \\
0 & \text{otherwise}
\end{cases}
\end{align*}
Note that such a prefix code $C'$ exists by Lemma~\ref{lem:Kraft-Inequality},
since $K_C(S) \le 1$ and
\begin{align*}
K_{C'}(S) - K_C(S)
&= \sum_{u\in A-B} \left( 2^{-l_{C'(u)}} - 2^{-l_{C(u)}} \right)
 + \sum_{u\in B-A} \left( 2^{-l_{C'(u)}} - 2^{-l_{C(u)}} \right)\\
&= \sum_{u\in A-B} 2^{-l_{C(u)}} \left( 2^{l_{C(u)}-l_{C'(u)}} - 1 \right)
 + \sum_{u\in B-A} 2^{-l_{C(u)}} \left( 2^{l_{C(u)}-l_{C'(u)}} - 1 \right)\\
&= (2^{-(j-i)} - 1) K_C(A-B) + (2^{-(i-j)} - 1) K_C(B-A) \\
&< (2^{-(j-i)} - 1) K_C(A-B) + (2^{-(i-j)} - 1) K_C(B-A) \\
&\ \ \ \ \ \ \ \ + (2^{-(j-i)} + 2^{-(i-j)} - 2) K_C(A \cap B) \\
&= (2^{-(j-i)} - 1) K_C(A) + (2^{-(i-j)} - 1) K_C(B) \\
&= (2^{-(j-i)} - 1) 2^{-i} + (2^{-(i-j)} - 1) 2^{-j} \\
&= 0 \end{align*}
where the inequality above
follows since $i$ and $j$ are nonnegative integers satisfying $j>i$,
which implies $2^{-(i-j)} \ge 2$ and $2^{-(j-i)} > 0$.
But then
\begin{align*}
E[l_{C'}(X)] - E[l_C(X)]
&= \sum_{u\in A-B} P(u)(l_{C'(u)} - l_{C(u)}) + \sum_{u\in B-A} P(u)(l_{C'(u)} - l_{C(u)})\\
&= (j-i) P(A-B) + (i-j) P(B - A)\\
&= (j-i) (P(A) - P(B)) \\
&< 0,
\end{align*}
which contradicts the optimality of $C$.
Thus $C$ is strongly monotone.

Now suppose $C$ is complete and strongly monotone,
and let $T$ be the code tree for $C$.
The completeness of $T$ implies that
every row of $T$ below the root 
has an even number of nodes in it since each node has a sibling.
The following iterative procedure
constructs a code tree $T'$
that is length equivalent to $T$
and whose node probabilities
are non-increasing from left to right in each row.
Begin by listing the leaves on the bottom row of $T$
in order of non-increasing probability
and combining them as siblings in pairs;
this is possible since there are an even number of such leaves on the row.
Then list the parent nodes just created
and the leaves in the second-lowest row of $T$
in order of non-increasing probability
and combine them as siblings in pairs;
again, this is possible for the same reason as in the previous step.
Continue this procedure from the bottom row to the top row,
until $T'$ is constructed.
The construction of $T'$ preserves which row its leaves came from in $T$,
so $T'$ is length equivalent to $T$.

Let $u$ be a node in row $i$ of $T'$
and let $v$ be a node in row $j$ of $T'$,
where $j > i$.
Let $C'$ be the prefix code whose code tree is $T'$,
and let $U,V \subseteq S$ be the sets of source symbols
corresponding to the leaf descendants of $u$ and $v$, respectively.
Since $C$ is complete,
Lemma~\ref{lem:length_equivalent_properties} implies $C'$ is complete.
Lemma~\ref{lem:complete_KS_1}
therefore implies $K_{C'}(U) = 2^{-i}$ and $K_{C'}(V) = 2^{-j}$,
and since $C'$ is length equivalent to $C$
we have
\begin{align*}
K_C(U) = K_{C'}(U) = 2^{-i} > 2^{-j} = K_{C'}(V) = K_C(V).
\end{align*}
Then since $C$ is strongly monotone,
$P(u) = P(U) \ge P(V) = P(v)$.
Therefore,
the list of nodes of $T'$ in raster-scan order,
beginning at the root node and moving down row-by-row,
left-to-right in each row,
has each node appearing adjacent to its sibling
and the node probabilities are non-increasing.
Since $C'$ is also complete,
$C'$ satisfies the sibling property, and so
Lemma~\ref{lem:SiblingProperty-iff-Huffman} implies that $C'$ is a Huffman code,
and thus $C'$ is optimal.
Since $C$ is length equivalent to $C'$,
Lemma~\ref{lem:length_equivalent_properties} implies
$C$ is optimal.
\end{proof}



\end{document}